\documentclass{article}

\usepackage[english]{babel}

\usepackage[letterpaper,top=2cm,bottom=2cm,left=3cm,right=3cm,marginparwidth=1.75cm]{geometry}

\usepackage{amsmath}
\usepackage{amsmath,amssymb,amsfonts}
\usepackage{graphicx}

\newtheorem{theorem}{Theorem}
\newtheorem{lemma}{Lemma}
\newtheorem{remark}{Remark}
\newenvironment{proof}{{\indent \indent \it Proof:\quad}}{\hfill $\blacksquare$\par}

\usepackage{subcaption} 

\usepackage[colorlinks=true, allcolors=blue]{hyperref}

\title{VA-CDH: A Variance-Aware Method to Optimize Latency for Caching with Delayed Hits}
\author{Bowen Jiang, Chaofan Ma, Duo Wang\\
Shanghai Jiao Tong University}

\begin{document}
\maketitle

\begin{abstract}
Caches are fundamental to latency-sensitive systems like Content Delivery Networks (CDNs) and Mobile Edge Computing (MEC). However, the delayed hit phenomenon—where multiple requests for an object occur during its fetch from the remote server after a miss—significantly inflates user-perceived latency. While recent algorithms acknowledge delayed hits by estimating the resulting aggregate delay, they predominantly focus on its mean value. We identify and demonstrate that such approaches are insufficient, as the real aggregate delay frequently exhibits substantial variance in the true production system, leading to suboptimal latency performance when ignored. 
Thus, we propose VA-CDH, a \underline{v}ariance-\underline{a}ware method to optimize latency for \underline{c}aching with \underline{d}elayed \underline{h}its.
 It employs a novel ranking function that explicitly incorporates both the empirically estimated mean and standard deviation of aggregate delay, allowing caching decisions to account for its variation. 
We derive the analytical distribution of aggregate delay under Poisson arrivals as a theoretical contribution, offering more statistical insight beyond the mean value. 
 Through the simulations conducted on synthetic and real-world datasets, we show that VA-CDH 
 reduces the total latency by  $1\%-6\%$  approximately compared to state-of-the-art algorithms.
\end{abstract}

\section{Introduction}
Cache is a fundamental technique employed across a wide array of computer and networking systems, including Content Delivery Networks (CDNs)\cite{cdn}, and storage systems \cite{storage}, to improve performance and user experience. 
When the requested object is found within the cache (\textbf{hit}), it can be served quickly with low latency. However, if the data is absent (\textbf{miss}), it must be fetched from a slower source, such as neighbor cache or the remote server. This necessary retrieval significantly increases response latency.
Traditionally, the primary objective of cache algorithms has been to maximize the cache hit ratio or byte hit ratio, thereby minimizing cache misses or the total cost of fetching data from remote server. Numerous algorithms, ranging from classic heuristics like Least Recently Used (LRU) \cite{LRU}  to more sophisticated approaches like Least Hit Density (LHD) \cite{lhd}, ADAPTSIZE \cite{adapt}, and various machine learning-based strategies \cite{lrb}, have been developed to optimize these metrics under different workloads and system constraints.

In high-throughput systems like CDNs and MEC, delayed hits significantly impact user-perceived latency. This occurs when frequent requests for an object arrive during the non-negligible latency of fetching it from the remote server after a miss. These subsequent requests, unlike true cache hits, are only served after being fetched back, thus accumulating delay. The widening gap between network throughput and transmission latency exacerbates this issue, increasing the number of requests affected during a miss fetch and amplifying the overall latency impact.
MAD \cite{mad} formalizes this phenomenon, introducing the aggregate delay metric defined as the sum of initial miss latency and latency of following delayed hits and using heuristic estimation based on historical patterns for cache decisions. Addressing heterogeneous fetch latencies and sizes, LAC \cite{atc} proposes a ranking function integrating both aggregate delay and mean latency per request. Recognizing limitations in estimation (like MAD \cite{mad} ) and the miss latency upper bounds, CALA \cite{cala} calculates rank by combining estimated average aggregate delay with the square of miss latency, modulated by a configurable parameter.

While these approaches represent significant advancements in acknowledging delayed hits, they primarily focus on estimating the mean aggregate delay. Our analysis, comparing the aggregate delay estimates of these algorithms against their real values, reveals a critical limitation: the real aggregate delay experienced often exhibits significant variance, deviating substantially from the mean estimates. Relying solely on mean values can be misleading, particularly for content with bursty or unpredictable request patterns. This oversight can lead to suboptimal caching decisions where objects with high average but also highly variable delay impacts might be incorrectly prioritized or evicted. We further derive the analytical distribution of aggregate delay under the assumption of Poisson arrivals, providing a formal basis for understanding its statistical properties beyond just the mean. 
Motivated by our analysis, we propose VA-CDH, a variance-aware method to optimize latency for caching with delayed hits.
In summary, our contributions are as follows:
\begin{itemize}
    \item We identify and demonstrate a critical limitation of existing delayed-hit algorithms: their failure to adequately account for the variance of aggregate delay.
    \item  We derive the distribution of aggregate delay under Poisson arrivals, providing fundamental insights into its statistical nature, including the mean and variance.
    \item We propose the VA-CDH algorithm, with a novel variance-aware ranking function and methods for online parameter estimation.
    \item Through simulations, we demonstrate that VA-CDH reduces the total latency by $1\%-6\%$ approximately compared to existing state-of-the-art algorithms.
\end{itemize}

\section{Problem Definition and Motivation}
\subsection{Setting}
We consider a cache system with a fixed capacity $C$ operating over a discrete time horizon $T$. At each time step $t$, a request $R_t$ arrives for an object $i$, drawn from a set of $N$ unique object types ( $N \gg C$ ). Each object type $i$ has an integer size $s_i$ and a deterministic fetch latency $z_i$. We assume that $\max _i s_i<C$.
If the requested object $R_t$ is present in the cache, it results in a cache hit, incurring zero latency. If object $i$ is not in the cache, a miss occurs and the corresponding latency is $z_i$. A fetch operation from a remote server is initiated, completing at time $t+z_i$.
Subsequent requests for the same object $i$ arriving at time $t^{\prime}$, where $t<t^{\prime} \leq t+z_i$, are considered delayed hits. These incur a latency equal to the remaining fetch time, $z_i-\left(t^{\prime}-t\right)$.
When the missed objects are fetched back, we need to decide which cached objects to evict to accommodate the newly fetched objects. The objective is to minimize the total latency incurred by all requests throughout the time horizon $T$.

\subsection{The Aggregate Delay and Motivation}\label{latency}

For an object $i$, we define its aggregate delay $D_i$ as the sum of the initial miss latency and the cumulative delay of subsequent delayed hits that occur during the fetch process. 
The time-dependent aggregate delay for object $i$ at time $t$ can be calculated as follows:
\begin{equation}\label{Aggdelay}
D_{i} 
=  z_i+\sum_{1 \leq \tau \leq z_i-1}\left(z_i-\tau\right)\mathbb{I}(R_{t+\tau}=i),
\end{equation}
where $\mathbb{I}$ is the indicator function and equals 1 if the arrival request at $t+\tau$ is object $i$, and 0 otherwise.
The aggregate
delay of object $i$ cannot be directly calculated in practice since it
requires the future information of the next $z_i$ time slots.
Thus several existing works have proposed methods to estimate it. For instance:
\begin{itemize}
    \item MAD \cite{mad}  assumes that
all past requests to object $i$ are missed, and then estimate the aggregate delay by calculating the average
aggregate delay per miss(AggDelay).

    \item LAC \cite{atc} estimates the aggregate delay using $D_i=\left(1+\frac{1}{1+\lambda_i z_i}\right) \frac{z_i}{2}$ under Poisson arrivals, where $\lambda_i$ is the arrival rate of object $i$.
    \item CALA \cite{cala} estimates aggregate delay as $D_i=(1-\gamma) \operatorname{AggDelay}+\gamma z_i^2$, where $\gamma$ is used
to adjust $\operatorname{AggDelay}$ and upper bound $z_i^2$.
\end{itemize}
We also note that \cite{tan2025asymptotically} incorporates the delay of eviction during fetch into their aggregate delay estimation. This inclusion stems from their consideration of a bypass scenario, where the user's missed request can be forwarded directly to the remote server instead of being handled by the local cache. Our work does not consider bypass, so their estimation method is not discussed here.
We evaluate three different estimation methods by comparing estimated aggregation latency against the real latency (calculated by equation Eq.~\eqref{Aggdelay}) for object with the highest arrival rate on four datasets: Wiki18 \cite{lrb}, Wiki19 \cite{lrb}, Cloud \cite{cloud}, and YouTube \cite{you}. The latency incurred upon a cache miss is modeled as a constant delay $L=10ms$ plus a component proportional to the object size.
The results are presented in Fig. \ref{fig:four_images}.

\begin{figure}[htbp] 
	\centering 

	\begin{subfigure}{0.48\columnwidth} 
		\centering 
		\includegraphics[width=\linewidth]{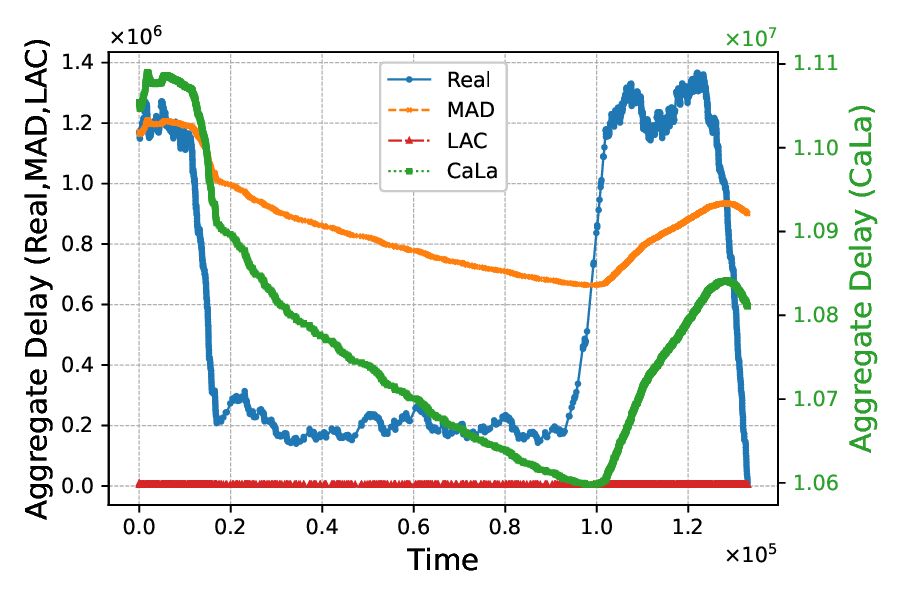} 
		\subcaption{\footnotesize{Wiki18}}
		\label{fig:wiki10}
	\end{subfigure}
	\hfill 
	\begin{subfigure}{0.48\columnwidth} 
		\centering
		\includegraphics[width=\linewidth]{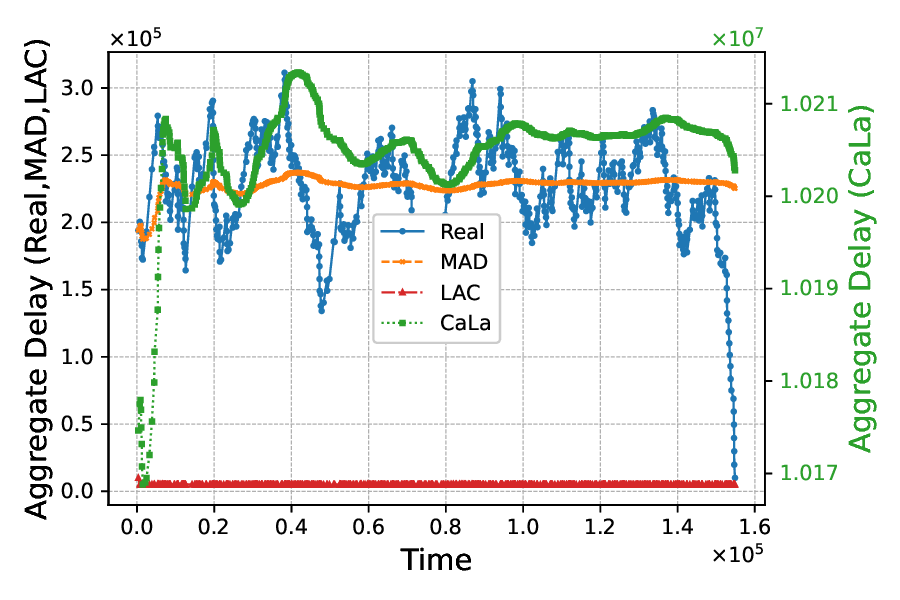} 
		\subcaption{\footnotesize{Wiki19}}
		\label{fig:wiki19}
	\end{subfigure}

	\vspace{\baselineskip} 

	\begin{subfigure}{0.48\columnwidth} 
		\centering
		\includegraphics[width=\linewidth]{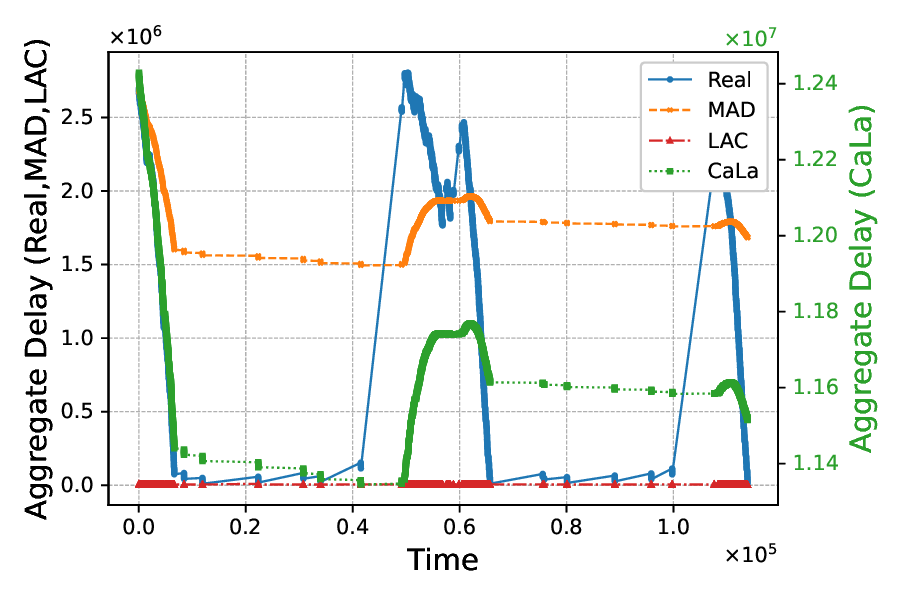} 
		\subcaption{\footnotesize{Cloud}}
		\label{fig:cloud}
	\end{subfigure}%
	\hfill 
	\begin{subfigure}{0.48\columnwidth} 
		\centering
		\includegraphics[width=\linewidth]{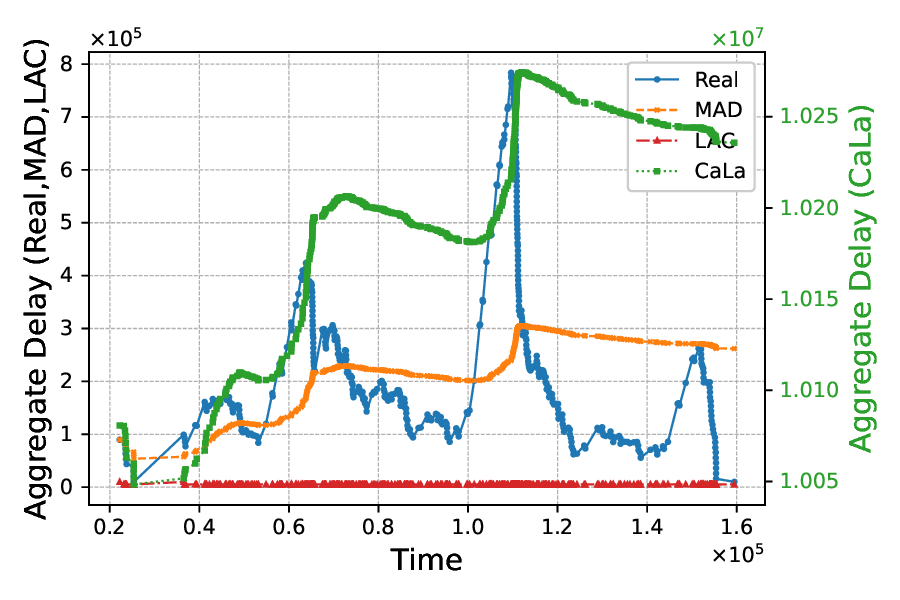} 
		\subcaption{\footnotesize{Youtube}}
		\label{fig:you}
	\end{subfigure}

	\caption{The estimated aggregate delay vs. real aggregate delay calculated by equation \ref{Aggdelay}.} 
	\label{fig:four_images} 
\end{figure}

Fig.\ref{fig:four_images} reveals that real aggregate delay fluctuates considerably, often with large magnitude changes. It also highlights the significant discrepancy between the real aggregate delay and the estimations from the existing three methods, which struggle to capture these dynamics effectively. This is particularly evident for Wiki19 dataset: while the real aggregate delay varies substantially, the MAD \cite{mad} remains largely flat, the LAC \cite{atc} consistently underestimates the delay, and the CALA \cite{cala} is significantly higher than the real value.

\section{Characteristics of Aggregate Delay}\label{texingfenxi}
The preceding discussion highlights the dynamic nature of aggregate delay, suggesting that characteristics capturing this variability should be considered in algorithm design. 
Instead of solely focusing on the average, it is valuable to study the distribution of the aggregate delay or other characteristics.

\subsection{The distribution of aggregate delay }

\begin{theorem} \label{prop:agg_delay_detailed_en}
Assume object $i$ arrives according to a Poisson process with rate $\lambda_i$, and the fetch time upon a miss is a constant $z_i$. The Probability Density Function (PDF) of the aggregate delay $D_i$ for object $i$ is a mixture distribution given by:
\begin{equation}
f(D_i) = e^{-\lambda_i z_i} \delta(D_i - z_i) + \sum_{k=1}^{\infty} P(N_i^z = k) \cdot f_{D_i|k}(D_i)
\label{eq:agg_delay_pdf_detailed_en}
\end{equation}
where $\delta(\cdot)$ is the Dirac delta function, $N_i^z$ is the random variable for the number of arrivals during the fetch period of length $z_i$, $P(N_i^z = k)$ is given as,
\begin{equation}
P(N_i^z = k) = \frac{e^{-\lambda_i z_i}(\lambda_i z_i)^k}{k!}, \quad k=0, 1, 2, \ldots
\label{eq:poisson_prob_detailed_en}
\end{equation}
and $f_{D_i|k}(D_i)$ is the conditional PDF of the aggregate delay given $k \ge 1$ arrivals:
\begin{equation}
f_{D_i|k}(D_i) = \frac{1}{z_i^k (k-1)!} \sum_{j=0}^{\lfloor D_i/z_i \rfloor - 1} (-1)^j \binom{k}{j} (D_i - (j+1)z_i)^{k-1},
\label{eq:conditional_pdf_detailed_en}
\end{equation}
defined for the range $z_i \le D_i \le (k+1)z_i$.
\end{theorem}

\begin{proof}
The proof uses the law of total probability, conditioning on $k=N_i^z$, the number of arrivals during the fetch interval of duration $z_i$. Since arrivals follow a Poisson process with rate $\lambda_i$, $N_i^z$ follows a Poisson distribution with parameter $\lambda_i z_i$, having the probability mass function (PMF) given by $P(N_i^z=k)$ in Eq.~\eqref{eq:poisson_prob_detailed_en}. We analyze the distribution of aggregate delay $D_i$ based on the value of $k$.

\emph{Case 1: $k=0$ (No arrivals).}
If $N_i^z=0$, no subsequent requests arrive during the fetch process. The only delay incurred is the initial miss latency $z_i$. Thus, $D_i = z_i$ when $k=0$. The conditional PDF is a point mass at $z_i$, represented by $f(D_i | N_i^z = 0) = \delta(D_i - z_i)$. Weighting this by the probability $P(N_i^z=0) = e^{-\lambda_i z_i}$ gives the first term in Eq.~\eqref{eq:agg_delay_pdf_detailed_en}.

\emph{Case 2: $k \ge 1$ arrivals.}
If $N_{i}^{z} = k \geq 1$, assume WLOG that the miss time is 0. Let the $k$ request arrival times within $(0, z_{i})$ be $\tau_{1}, \ldots, \tau_{k}$. Given $k$ Poisson arrivals in $(0, z_{i})$, these times $\tau_{j}$ are i.i.d. U$(0, z_{i})$. The aggregate delay is thus:
$D_{i} = z_{i} + \sum_{j=1}^{k} (z_{i} - \tau_{j})$.

Let $V_j = z_i - \tau_j$. Since $\tau_j \sim U(0, z_i)$, by linear transformation, $V_j$ is also distributed as $\text{Uniform}(0, z_i)$.
Let $S_k = \sum_{j=1}^k V_j$. Then $D_i | (N_i^z=k) = z_i + S_k$. The problem reduces to finding the PDF of $S_k$, the sum of $k$ i.i.d. $U(0, z_i)$ variables.
This distribution is known as the scaled Irwin-Hall distribution \cite{IH}. Let $X_j = V_j / z_i \sim U(0, 1)$. Then $X_{(k)} = \sum_{j=1}^k X_j$ follows the standard Irwin-Hall distribution with a known PDF.
Since $S_k = z_i X_{(k)}$, we can find the PDF of $S_k$ using the change of variables formula $f_{S_k}(s) = f_{X_{(k)}}(s/z_i) \cdot |dx/ds| = \frac{1}{z_i} f_{X_{(k)}}(s/z_i)$. Plugging in the standard Irwin-Hall PDF gives:
\begin{equation}
    f_{S_k}(s) = \frac{1}{z_i^k (k-1)!} \sum_{j=0}^{\lfloor s/z_i \rfloor} (-1)^j \binom{k}{j} (s - jz_i)^{k-1},
\end{equation}
for $0 \le s \le k z_i$.

Finally, we find the conditional PDF of $D_i$ given $k$, denoted $f_{D_i|k}(d)$, using the transformation $D_i = z_i + S_k$ (so $s = d - z_i$ and $|ds/dd|=1$).
\begin{align*}
f_{D_i|k}(d) &= f_{S_k}(d-z_i) \cdot |ds/dd| \\
&= \frac{1}{z_i^k (k-1)!} \sum_{j=0}^{\lfloor (d-z_i)/z_i \rfloor} (-1)^j \binom{k}{j} ((d-z_i) - jz_i)^{k-1} \\
&= \frac{1}{z_i^k (k-1)!} \sum_{j=0}^{\lfloor d/z_i \rfloor - 1} (-1)^j \binom{k}{j} (d - (j+1)z_i)^{k-1}. \label{eq:derivation_step_detailed_en} \tag{$\dagger$}
\end{align*}
The range for $s$ ($0 \le s \le kz_i$) implies the range for $d = z_i + s$ is $z_i \le d \le (k+1)z_i$. The summation limit $\lfloor (d-z_i)/z_i \rfloor$ becomes $\lfloor d/z_i \rfloor - 1$. The expression $(\dagger)$ matches Eq.~\eqref{eq:conditional_pdf_detailed_en}.

\emph{Combining Cases:}
Summing the conditional results for all possible values of $k \ge 1$, weighted by their Poisson probabilities $P(N_i^z=k)$, and adding the term for $k=0$, yields the overall PDF for $D_i$ as stated in Eq.~\eqref{eq:agg_delay_pdf_detailed_en}.
\end{proof}

\subsection{The mean and variance of aggregate dealy}
We derive the mean and variance of the $D_i$ under the Poisson arrival assumption in Theorem \ref{meanandvar}. This derivation first requires determining the conditional mean and variance of aggregate delay given a fixed number of arrivals during the fetch process, $k$, as established in Lemma\ref{mean}.

\begin{lemma}\label{mean}
  With Poisson arrivals and the miss latency for object $i$ is $z_i$, the conditional mean and variance of the aggregate delay $D_i$ given $k$ arrivals are:
 $E\left[D_i \mid N_i^z=k\right]=z_i\left(1+\frac{k}{2}\right)$, $\operatorname{Var}\left(D_{i} \mid N_i^z=k\right)=\frac{k z_i^2}{12}$.  
\end{lemma}
Lemma \ref{mean} is based on the properties of mean and variance, and the characteristics of uniform distribution. Due to limited space, its proof is not detailed here. Based on Lemma \ref{mean}, we can derive the mean and variance of the aggregation delay.

\begin{theorem}\label{meanandvar}
Assume object $i$ arrives according to a Poisson process with rate $\lambda_i$ and the miss latency is $z_i$. Based on Lemma \ref{mean}, the mean and variance of $D_i$ are:
\begin{equation}
    E[D_i] = z_i \left(1 + \frac{\lambda_iz_i}{2}\right),\operatorname{Var}(D_i) = \frac{z_i^3 \lambda_i}{3}
\end{equation}

\end{theorem}

\begin{proof}
We use the Law of Total Expectation and the Law of Total Variance.
Recall that for $N_i^z \sim \text{Poisson}(\lambda_iz_i)$, $E[N_i^z] = \lambda_iz_i$ and $\operatorname{Var}(N_i^z) = \lambda_iz_i$.

By the Law of Total Expectation, we have:

\begin{equation}
\begin{split}
        E[D_i] &= E_{N_i^z}[E[D_i | N_i^z]]= E_{N_i^z}\left[z_i\left(1 + \frac{N_i^z}{2}\right)\right]  \\
&= z_i\left(1 + \frac{E_{N_i^z}[N_i^z]}{2}\right)= z_i\left(1 + \frac{\lambda_iz_i}{2}\right) 
\end{split}
\end{equation}

By the Law of Total Variance, we have:
\begin{equation}
\begin{split}
\operatorname{Var}(D_i) &= E_{N_i^z}[\operatorname{Var}(D_i | N_i^z)] + \operatorname{Var}_{N_i^z}(E[D_i | N_i^z]) \\
&= E_{N_i^z}\left[\frac{N_i^z z_i^2}{12}\right] + \operatorname{Var}_{N_i^z}\left(z_i\left(1 + \frac{N_i^z}{2}\right)\right) \\
&= \frac{z_i^3 \lambda_i}{3},
\end{split}
\end{equation}
where the last step is due to $E[N_i^z] = \lambda_iz_i$, $\operatorname{Var}(N_i^z) = \lambda_iz_i$ and $\operatorname{Var}(z_i)=0$.
\end{proof}

\begin{remark}
\textbf{(Gaussian Approximation).}
    When the number of subsequent delayed hits $k$ is large, the conditional distribution of aggregate delay, $f_{D_i|k}(d)$, can be well approximated by a Gaussian distribution. This follows because $D_i | (N_i^z=k) = z_i + S_k$, and for large $k$, $f_{S_k}$ itself is approximately Gaussian due to the Central Limit Theorem applied to the sum of uniform random variables. Since $D_i | (N_i^z=k)$ is a linear transformation of $S_k$, it retains an approximately Gaussian form. Based on Proposition.\ref{mean}, for large $k$, we can approximate the distribution as:
$
f_{D_i|k}(d) \approx \mathcal{N}\left(z_i(1+k / 2), \frac{k z_i^2}{12}\right).
$

\end{remark}

\begin{remark}
\textbf{(Parameter dependency).}
    Theorem \ref{meanandvar} considers the aggregate delay when the miss latency, $z_i$, is constant. 
  The mean aggregate delay, $E[D_i]$, increases linearly with the arrival rate $\lambda_i$. Furthermore, the presence of the $\lambda_i z_i^2$ term shows that the impact of the arrival rate is amplified quadratically by the fetch latency, as longer fetches provide more opportunity for costly delayed hits. The variance, $Var(D_i)$, grows linearly with the arrival rate $\lambda_i$, indicating greater variability when subsequent requests are more frequent during the fetch. Most notably, the variance exhibits a strong cubic dependence ($z_i^3$) on the fetch latency $z_i$, highlighting that the variability  of the aggregate delay experienced during a miss escalates rapidly as fetch latency increases.
\end{remark}

\section{Variance-Aware Methond}\label{four}

\subsection{Variance-Aware Rank Function}
In order to consider the variation of aggregation delay in algorithm design, we design the following ranking function:
\begin{equation}\label{rank}
    f_i=\frac{E[D_i] + \omega \sigma[D_i]}{R_is_i},
\end{equation}
where $E[D_i]$ is the mean value, $\sigma[D_i]$ is the standard deviation,
 $\omega >0 $ controls the algorithm’s sensitivity to the aggregate delay variability, $R_i$ is time interval remaining until next request, and $s_i$ is the size of object $i$.
 When $\omega=0 $, our rank function degenerates to only considering the mean of the aggregate delay, just like MAD \cite{mad}.
Objects with a higher rank $f_i$ are considered more valuable and less likely to be evicted. Based on the mean and variance of aggregation delay under deterministic miss latency discussed in section \ref{texingfenxi}, we can obtain the following ranking function,
\begin{equation}\label{rank2}
    f_i=\frac{z_i \left(1 + \frac{\lambda_iz_i}{2}\right)  + \omega z_i\sqrt{\frac{ \lambda_iz_i}{3} } }{R_is_i},
\end{equation}
where $z_i$ is miss latency, $\lambda_i$ is the arrival rate for object $i$.

The idea of combining mean and variance (or standard deviation) for decision-making under uncertainty is fundamental and appears in various fields. In financial portfolio theory \cite{MPT}, investors evaluate assets based on expected return (mean) and risk (variance/standard deviation). Rational investors aim to minimize risk for a given level of return or maximize return for a given level of risk. Utility functions often incorporate both, e.g., $U = E[R] - \frac{A}{2}\sigma^2[R]$, where $R$ is the return and $A$ is the risk aversion coefficient.
Structurally, our ranking function, which adds a term related to standard deviation to the mean, shares a similar idea with the Upper Confidence Bound (UCB) algorithm in bandit problems \cite{bandit} that also combines an estimate with an uncertainty bonus.

\subsection{VA-CDH Algorithm}
To determine the rank of cached objects when eviction is needed, we need estimations for both the residual time $R_i$ and the arrival rate $\lambda_i$. Estimating $R_i$ is a well-explored area, with many hit ratio focused algorithms effectively acting as estimators; therefore, we leverage existing methods like LRU for this aspect. Specifically, we take the time interval between the current time $t$ and the last arrival time of object $i$ as the estimated value of $R_i$. We estimate the request arrival rate $\lambda_{i}$ for each content $i$ by approximating its corresponding mean inter-arrival time, $\bar{X}^{i}$. Denoting the $j$-th observed inter-arrival time for object $i$ as $X_{j}^{i}$, we calculate the sample mean from $K$ observations: $\hat{\bar{X}}^{i}=\frac{1}{K}\sum_{j=1}^{K}X_{j}^{i}$. It is well-established that $\hat{\bar{X}}^{i}$ is an unbiased estimator of the true mean inter-arrival time $\bar{X}^{i}=\frac{1}{\lambda_{i}}$.
To manage overhead and adapt to changing request patterns, tracking the entire arrival history is impractical. Since production workloads are known to be highly dynamic and non-stationary, we utilize a sliding learning window with size $S$, estimating the parameters $R_i$ and $\lambda_i$ based only on the $S$ most recent arrivals within this window.


\textbf{VA-CDH Algorithm:} 
When a request arrives, we add it to the learning window and update the characteristics of this request (inter-arrival time, latest arrival time). 
When we need to make evictions, we use the estimated residual time $R_i$ and arrival rate $\lambda_i$ to calculate the rank of each cached object based on equation Eq.~\eqref{rank2} and evict the objects with the smallest ranks to store the objects fetched from the remote server.
When the training window is full, we use the data within the training window to re-estimate the arrival rate.

\section{Simulation Results}\label{five}
We consider the baseline algorithms, including  LRU \cite{LRU}, LAC \cite{atc},
CALA \cite{cala}, LHD\cite{lhd}, LRB\cite{lrb}, ADAPTSIZE \cite{adapt},
LHDMAD \cite{mad} and LRUMAD\cite{mad}.
Notably, the algorithms LAC \cite{atc}, CALA \cite{cala}, LHDMAD \cite{mad}, and LRUMAD \cite{mad} address the caching problems with delay hits, while ADAPTSIZE \cite{adapt}, LHD \cite{lhd}, and LRB \cite{lrb} focus on maximizing cache hit ratio without delayed hits.
The performance of algorithm A is measured by the latency improvement relative to LRU, which is defined as the ratio of the latency difference between LRU and algorithm A to the latency of LRU \cite{cala}. In subsequent experiments, we set the window size to 10K and $\omega=1$. The impact of these parameters on the performance of VA-CDH is discussed in subsection \ref{mingan}.


\subsection{Synthetic dataset}\label{hechengshiyan}

\begin{figure}[h]
\centering
\includegraphics[width=3in]{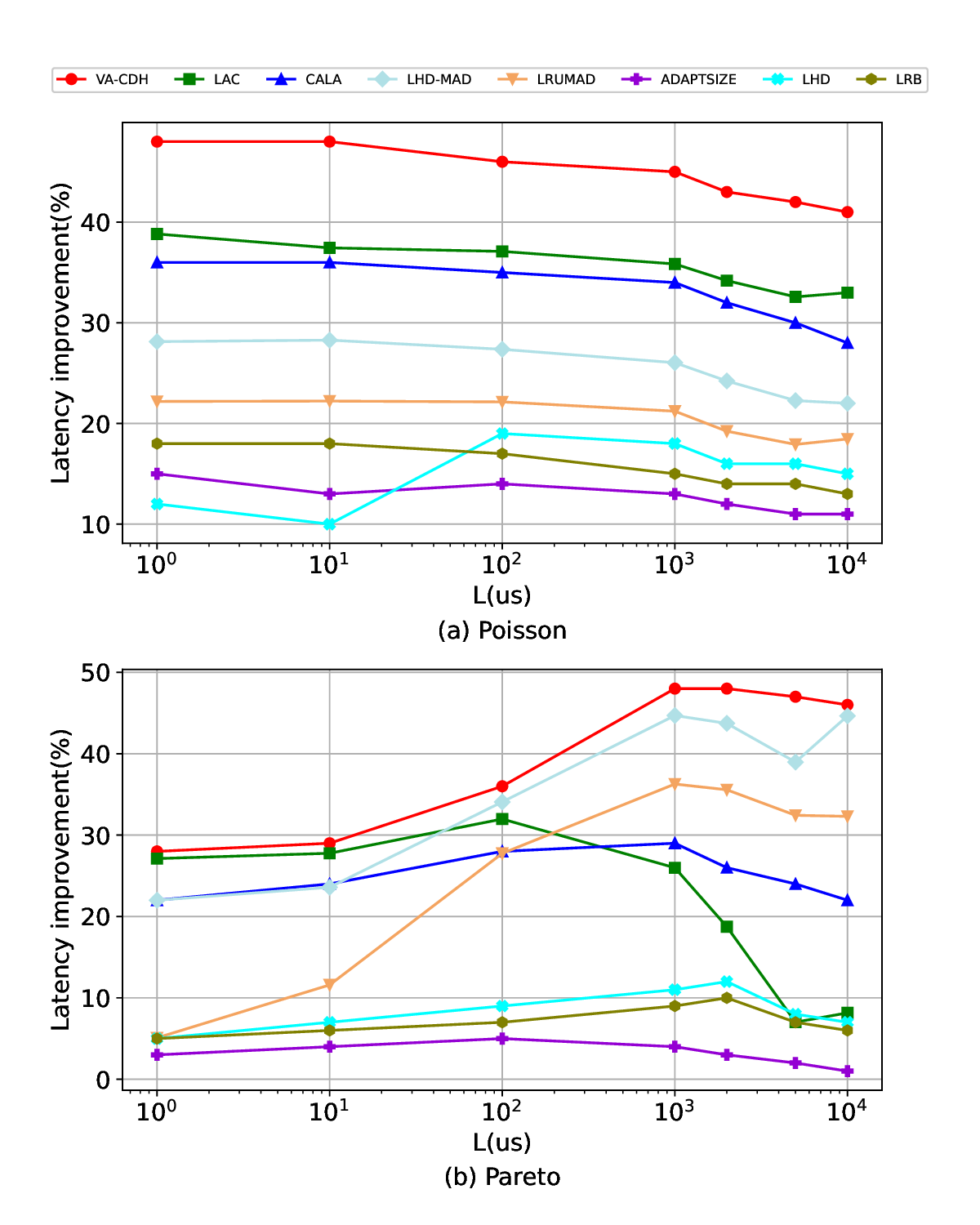}
\caption{Comparison of latency improvement of VA-CDH and SOTAs under the
synthetic dataset ($C=500\text{MB}$).
}
\label{hechengqueding}
\end{figure}

We first validate the performance of VA-CDH on a synthetic dataset. This dataset consists of 100K requests targeting 100 distinct objects, with object popularity following a Zipf distribution. The object sizes are integers uniformly distributed in the range $[1 \text{MB}, 100 \text{MB}]$. We evaluate our VA-CDH algorithm under two distinct request arrival processes: one governed by a Poisson process and the other by a Pareto distribution. 
The latency incurred upon a cache miss is as discussed in subsection \ref{latency} .
Fig. \ref {hechengqueding} presents the results, showing that VA-CDH outperforms state-of-the-art (SOTA) algorithms across various miss latency settings. Notably, this superior performance holds even when the arrival process follows a Pareto distribution, which violates our Poisson assumption.

\subsection{Real world traces}\label{shiji}


\begin{figure*}[t]
\centering
{\includegraphics[width=6.5in]{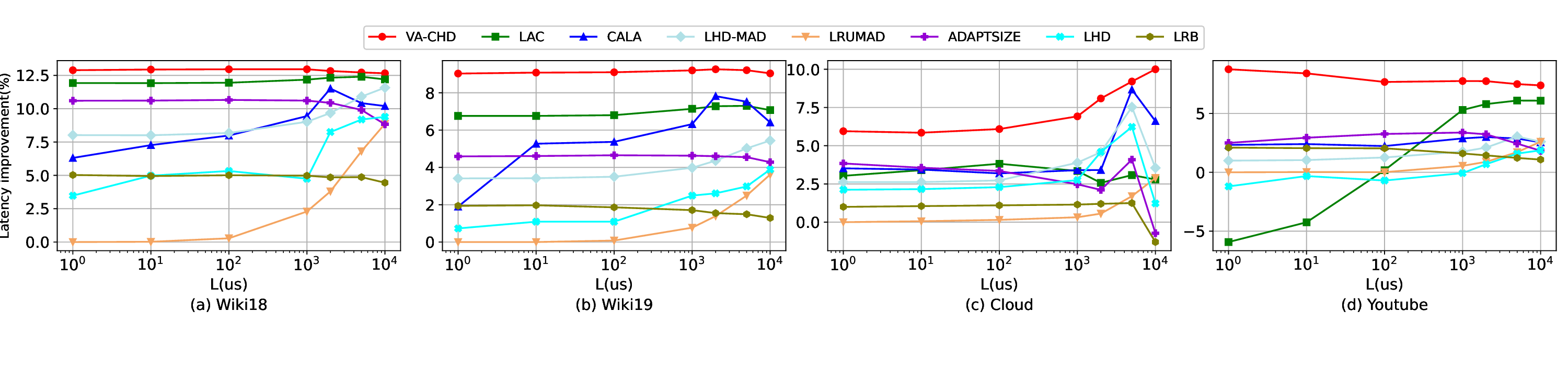}}%
\label{fig_first_case}
\caption{Comparison of latency improvement of VA-CDH and SOTAs using a 256GB cache.}
\label{256queding}
\end{figure*}

We evaluate our algorithm using four real-world traces: Wiki2018 \cite{lrb}, Wiki2019 \cite{lrb}, Cloud \cite{cloud}, and YouTube \cite{you}. 
Fig. \ref {256queding} compares our algorithm's latency improvement against SOTA methods using a 256 GB cache across various fetch latency settings. Our algorithm consistently outperforms the SOTA approaches under these conditions. Specifically, our algorithm achieves latency reduction of approximately $1\%$ on Wiki2018, around $3\%$ on Wiki2019, between $2\%-4\%$ on Cloud, and $1\%-6\%$ on YouTube. This performance advantage is primarily attributed to our ranking function, which incorporates the variance of delayed hits to prioritize eviction candidates more effectively.

\subsection{Sensitivity Analysis}\label{mingan}

\begin{figure}[h]
\centering
\includegraphics[width=2.5in]{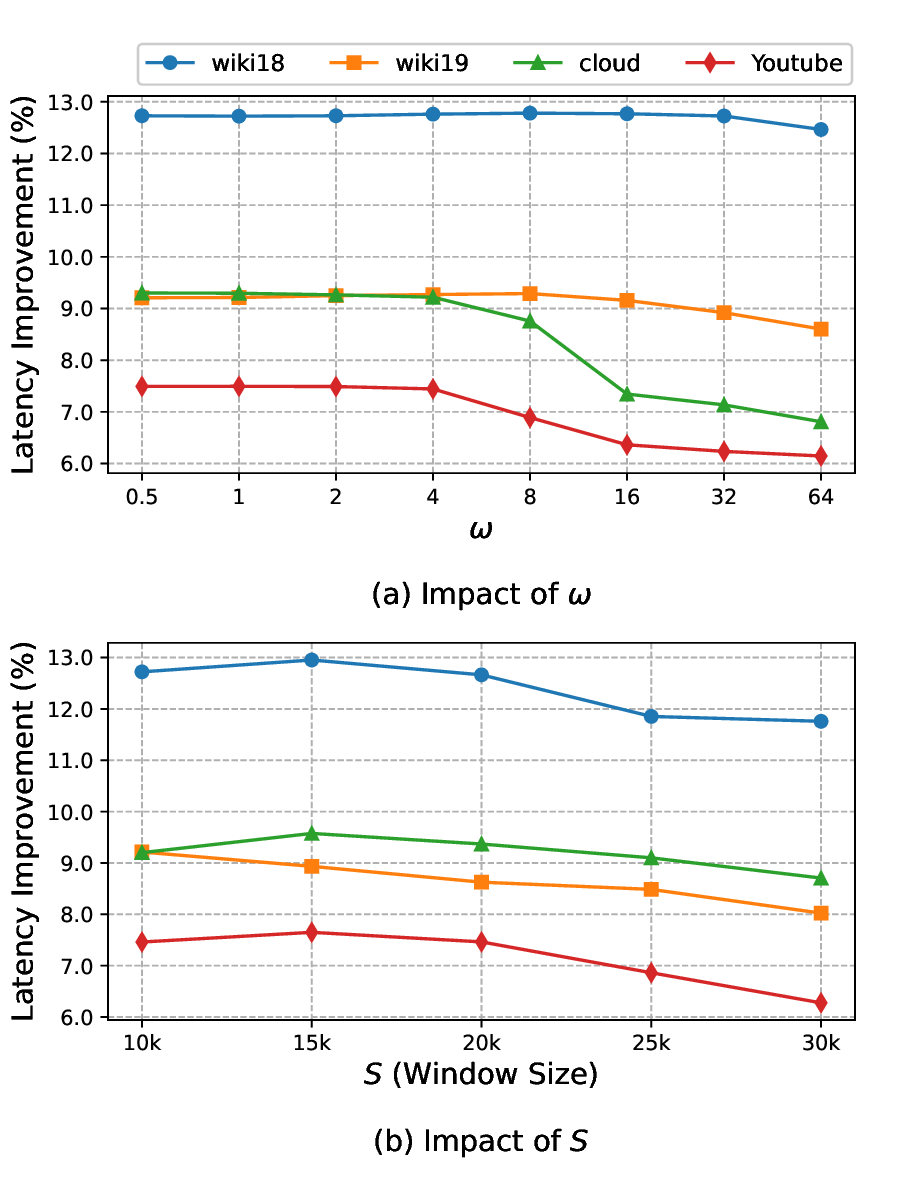}
\caption{The impact of hyperparameters on latency improvement. (a) fixed $S=10K$, (b) fixed $\omega=1$.
}
\label{canshu}
\end{figure}

The performance of our algorithm depends on two key hyperparameters: the sliding window size $S$ and the control parameter $\omega$. 
The choice of the parameter $\omega$ involves a trade-off: if $\omega$ is too small, the algorithm may not sufficiently capture the impact of variance. Conversely, if $\omega$ is too large, the variance term can dominate the mean, potentially impairing the effectiveness of the priority determination. Similarly, the window size $S$ selection is critical: an overly small window prevents the algorithm from adequately learning request arrival characteristics, while an excessively large window hinders its ability to adapt effectively to changes in the request arrival process. Consequently, we perform a sensitivity analysis to evaluate their impact, and the simulation results are shown in Fig. \ref{canshu}.
The experimental results indicate that our algorithm outperforms current SOTA methods across these different parameter settings. 
This demonstrates the robustness of our algorithm.

\section{Conclusion}\label{six}
In this paper, we address the delayed hit caching problem, highlighting that existing methods critically overlook the variance of aggregate delay. We propose a novel variance-aware algorithm that estimates and incorporates both the mean and standard deviation of aggregate delay into its online ranking function. We show that our approach achieves superior performance over current state-of-the-art delayed-hit algorithms through simulations.

\end{document}